\documentclass[submission,copyright,creativecommons]{eptcs}
\providecommand{\event}{PLACES 2010}
\usepackage{color} 
\usepackage{amsmath, amsthm, amssymb, stmaryrd}
\usepackage{xspace}
\usepackage{listings}
\newcommand{\eg}{\textit{e.g.},\@\xspace}
\newcommand{\ie}{\textit{i.e.,}\@\xspace}
\newcommand{\etal}{\textit{et al.}\@\xspace}

\definecolor{annotationcolor}{rgb}{0.8,0.8,0.8}
\newcommand{\annotation}[1]{{\colorbox{annotationcolor}{$#1$}}}

\newcommand{\xtenclocks}{$\textsf{X10}\!\mid_{\textsf{clocks}}$\xspace}

\newcommand{\kw}[1]{\text{\lstinline|#1|}}
\newcommand{\asynck}{\kw{async}}
\newcommand{\nextk}{\kw{next}}
\newcommand{\resumek}{\kw{resume}}
\newcommand{\dropk}{\kw{drop}}
\newcommand{\clockk}{\kw{clock}}
\newcommand{\newClockk}{\kw{makeClock}}
\newcommand{\ink}{\kw{in}}

\newcommand{\finishk}{\kw{finish}}
\newcommand{\letk}{\kw{let}}

\newcommand{\unitk}{\kw{unit}}
\newcommand{\joink}{\kw{join}}

\newcommand{\join}[1]{\joink\ {#1}}
\newcommand{\async}[2]{\asynck\ {#1}\ {#2}}
\newcommand{\resume}[1]{\resumek\ {#1}}
\newcommand{\drop}[1]{\dropk\ {#1}}
\newcommand{\newClock}{\newClockk}

\newcommand{\lets}[3]{\letk\ {#1}={#2}\ \ink\ {#3}}
\newcommand{\finish}[1]{\finishk {#1}}
\newcommand{\dom}{\operatorname{dom}}

\newcommand{\state}[2]{ {#1}; {#2}}

\newcommand{\activity}[3]{({#1}, {#2}, {#3})}
\newcommand{\heapClock}[3]{\langle{#1}, {#2}, {#3}\rangle}

\newcommand{\unitV}{()}

\newcommand{\mkRrule}[1]{\textsc{R-#1}}
\newcommand{\mkTrule}[1]{\textsc{T-#1}}

\newcommand{\mkErule}[1]{\textsc{E-#1}}
\newcommand{\mkWFrule}[1]{\textsc{WF-#1}}

\newcommand{\Rletval}{\mkRrule{let-val}}
\newcommand{\Rlet}{\mkRrule{let}}

\newcommand{\Rmake}{\mkRrule{make}}

\newcommand{\Rasync}{\mkRrule{async}}
\newcommand{\Rresume}{\mkRrule{resume}}

\newcommand{\Rnext}{\mkRrule{next}}
\newcommand{\Rdrop}{\mkRrule{drop}}
\newcommand{\Ractivity}{\mkRrule{activity}}
\newcommand{\Rfinish}{\mkRrule{finish}}
\newcommand{\Rjoin}{\mkRrule{join}}

\newcommand{\Easync}{\mkErule{async}}

\newcommand{\Eresume}{\mkErule{resume}}
\newcommand{\Edrop}{\mkErule{drop}}

\newcommand{\EactSet}{\mkErule{act-set}}
\newcommand{\Eact}{\mkErule{act}}

\newcommand{\EnextI}{\mkErule{next1}}
\newcommand{\EnextII}{\mkErule{next2}}

\newcommand{\WFstate}{\mkWFrule{state}}
\newcommand{\WFactSet}{\mkWFrule{act-set}}
\newcommand{\WFactSetE}{\mkWFrule{act-set-e}}
\newcommand{\WFheap}{\mkWFrule{heap}}
\newcommand{\WFheapE}{\mkWFrule{heap-e}}
\newcommand{\WFactClock}{\mkWFrule{act-clock}}
\newcommand{\WFactClockE}{\mkWFrule{act-clock-e}}

\newcommand{\update}[2]{\{{#1}\colon {#2}\}}
\newcommand{\subst}[2]{[{#1}/{#2}]}

\newcommand{\err}{\mathrm{Error}}
\newcommand{\nerr}{\notin \mathrm{Error}} 
\newcommand{\restrictedto}[1]{\!\mid_{#1}}
\newcommand{\load}{\operatorname{load}}

\newcommand{\Tlet}{\mkTrule{let}}

\newcommand{\Tasync}{\mkTrule{async}}
\newcommand{\Tmake}{\mkTrule{make}}
\newcommand{\Tresume}{\mkTrule{resume}}
\newcommand{\Tnext}{\mkTrule{next}}
\newcommand{\Tdrop}{\mkTrule{drop}}
\newcommand{\Theap}{\mkTrule{heap}}

\newcommand{\TcheckView}{\mkTrule{view}}
\newcommand{\Tact}{\mkTrule{act}}
\newcommand{\TactSet}{\mkTrule{act-set}}

\newcommand{\TclockSeq}{\mkTrule{clock-seq}}

\newcommand{\Tstate}{\mkTrule{state}}
\newcommand{\Tunit}{\mkTrule{unit}}

\newcommand{\TclockR}{\mkTrule{clock-ref}}
\newcommand{\Tvar}{\mkTrule{var}}
\newcommand{\Twfc}{\mkTrule{wf-c}}
\newcommand{\Twfu}{\mkTrule{wf-u}}

\newcommand{\Tfinish}{\mkTrule{finish}}
\newcommand{\Tjoin}{\mkTrule{join}}
\newcommand{\Tvalue}{\mkTrule{value}}

\newcommand{\clockT}[1]{\clockk({#1})}
\newcommand{\unitT}{\unitk}

\newcommand{\pad}{\;\;}
\newcommand{\Space}[1]{\pad{#1}\pad}
\newcommand{\grmeq}{\Space{::=}}

\newcommand{\grmor}{\;\mid\;}
\newcommand{\grmin}{\quad \;} 
\newcommand{\eqdef}{\triangleq}

\newtheorem{thm}{Theorem}
\newtheorem{lem}[thm]{Lemma}

\newtheorem{cor}[thm]{Corollary}

\lstdefinelanguage{x10clocks}{
  morekeywords={async, makeClock, drop, finish, next, resume, let, in,
     join, unit, clock, clocked},
  morecomment=[s]{\{-}{-\}},
  morecomment=[l]//,
  moredelim=[is][\emph]{'}{'},
  flexiblecolumns=true,
   literate=
          {e0}{$e_0$}2 {e1}{$e_1$}2
          {e2}{$e_2$}2 {e3}{$e_3$}2
          {e4}{$e_4$}2 {e5}{$e_5$}2
          {e6}{$e_4$}2 {e7}{$e_5$}2
          {a1}{$a_1$}2 {a2}{$a_2$}2
          {a3}{$a_3$}2 {a4}{$a_4$}2
          {alpha}{$\alpha$}1
          {emptyset}{$\emptyset$}1
}

\title{Types for X10 Clocks}

\def\titlerunning{X10 clocks}

\author{
Francisco Martins
\institute{LaSIGE \& University of Lisbon\\Portugal}
\email{fmartins@di.fc.ul.pt}
\and
Vasco T. Vasconcelos
\institute{LaSIGE \& University of Lisbon\\Portugal}
\email{vv@di.fc.ul.pt}
\and
Tiago Cogumbreiro
\institute{LaSIGE \& University of Lisbon\\Portugal}
\email{cogumbreiro@di.fc.ul.pt}
}

\def\authorrunning{F. Martins, V. T. Vasconcelos \& T. Cogumbreiro}

\begin{document}
\maketitle
\begin{abstract}
  X10 is a modern language built from the ground up to handle future
  parallel systems, from multicore machines to cluster configurations.
  We take a closer look at a pair of synchronisation mechanisms:
  finish and clocks. The former waits for the termination of parallel
  computations, the latter allow multiple concurrent activities to
  wait for each other at certain points in time.
  In order to better understand these concepts we study a type system
  for a stripped down version of X10.
  The main result assures that well typed programs do not run into the
  errors identified in the X10 language reference, namely the
  \texttt{ClockUseException}.
  The study will open, we hope, doors to a more flexible utilisation
  of clocks in the X10 language.
\end{abstract}


\section{Introduction}
\label{sec:motivation}

New high-level concurrency primitives are needed more than ever, now
that multicore machines lay on our desks and laps.
One such primitive is \emph{clocks}, a generalisation of barriers
introduced in the X10 programming language~\cite{charles.etal:x10}.
%
%
Barriers are a collective synchronisation mechanism common in Single
Program Multiple Data (SPMD)
programs~\cite{darema:spmd-model,tseng:compiler-opt-barrier}.
Distinct from synchronisation mechanisms like locks and monitors that
let the programmer think about the access to a resource, barriers
allow reasoning about process workflow.
Clocks are a sophisticated form of barriers that feature dynamic sets
of participants for more dynamic programming
paradigms, and a two-phase synchronisation (or fuzzy barrier) for
improved processor utilisation~\cite{gupta:fuzzy-barrier}.
The construct is integrated in the X10 language with a promise that it
cannot introduce deadlocks.
Another primitive, finish, causes an activity to block (\ie to suspend
its execution) until all its sub-activities have completed.
Dynamic, unbounded spawning of activities and the finish construct
enable fork/join parallelism. 
The fifth version of language Cilk introduced a specialised
``work-stealing'' scheduler algorithm that takes advantage of the
fork/join model to yield highly efficient language~\cite{frigo:cilk5}.
This style of parallel programming and its run-time system (the
work-stealing scheduler) were then incorporated in mainstream
languages, \eg the fork/join framework proposed for
Java~7~\cite{lea:java-fork-join}.

Even though the X10 language specification~\cite{saraswat:x10-report}
provides a clear, plain English, description of the intended semantics
(and properties) of the language, and a formalisation of the
semantics~\cite{saraswat.jagadeesan:concurrent-clustered-programming}
allows to prove a deadlock freedom theorem, we decided to investigate
a simpler setting in which similar results could be obtained. The aim
is not only to obtain a progress property for typable programs based
on a simple type system, but also to hopefully provide for clock-safe
extensions of the X10 language itself.

Towards this end, we have stripped X10 from most of its features,
ending up with a simple concurrent language equipped with finish
and the full functionality of X10 clocks, which we call ``X10
restricted to clocks,'' \xtenclocks{} for short.
For this language we have devised a simple operational semantics with
thread (or activity as called in X10) local and global views of (heap
allocated) clocks. We have also crafted a simple type system, based on
singleton types, drawing expertise from previous work on low-level
programming
languages~\cite{vasconcelos.martins.cogumbreiro:type-inference-mil}.
Typable programs are exempt from clock related errors;
we conjecture that typable programs enjoy a form of progress property.

Saraswat and Jagadeesan study the X10 programming model, by presenting
a formal model of the language that includes clocks, async/finish,
conditional atomic blocks, and a hierarchic shared
memory~\cite{saraswat.jagadeesan:concurrent-clustered-programming}.
The authors formalise the semantics of the language with a small-step
operational semantics, define a bisimulation, and establish that X10
programs without conditional atomic blocks do not deadlock.
Lee and Palsberg present a core model for X10, an imperative language
augmented with async/finish and atomic blocks suited for
inter-procedural analysis through type
inference~\cite{lee.palsberg:fx10}.
The authors present an operational semantics based
on~\cite{saraswat.jagadeesan:concurrent-clustered-programming} and a
type system that identifies may-happen-parallelism, further explored
in~\cite{agarwal:mhp-x10}.
Java features, since version 5, cyclic barriers,
\texttt{java.util.concurrent.CyclicBarrier}. Unlike X10 clocks these
barriers have a fixed set of participants, defined at initialisation
time.

This paper constitutes an archival version (post-proceedings) of a
workshop paper produced in late 2009.
Three of the language extensions proposed at that time have been
incorporated in the X10 language
definition~\cite{saraswat:x10-report}, further discussed in
Section~\ref{sec:type-system}.
They are:
\begin{itemize}
\item Aliasing, introduced in version 2.01, January 2010. Prior
  versions imposed a restriction on how clock values could be aliased,
  ``The initializer for a local variable declaration of type
  \texttt{Clock} must be a new clock expression. Thus X10 does not
  permit aliasing of clocks'' (page 123, version 2.00).
\item Clocks can be transmitted to a spawned activity if they are
  created in the scope of the enclosing finish, introduced in version
  2.05, July 2010.
  Before we could read, ``While executing \texttt{S} [the body of a
  finish], an activity must not spawn any \texttt{clocked}
  asyncs''\footnote{In X10 a \texttt{clocked} async corresponds to an
    activity registered with at least one clock upon
    creation.}  (page 141, version 2.04). Shirako
  \etal included this extension even prior to our
  work~\cite{shirako.peixotto.etal:phasers}.

\item Resume state inheritance, introduced in version 2.05, July
  2010. Resumed clocks can be transmitted to forked activities.
  Prior to version 2.05 we could read, ``It is a static error if any
  activity has a potentially live execution path from a
  \texttt{resume} statement on a clock~\texttt{c} to a async spawn
  statement (which registers the new activity on~\texttt{c}) unless
  the path goes through a \texttt{next} statement''  (page 153, version
  2.04).

\end{itemize}

In summary, the contributions of this work are:
\begin{itemize}
\item a simple operational semantics for activities, finish, and
  clocks that allows to better understand these constructs,
\item a type systems allowing to prove safety and progress
  properties (alternative to the constraint-based
  system~\cite{saraswat.jagadeesan:concurrent-clustered-programming}),
  and
\item the promise of a more flexible utilization of the clock
  constructs.

\end{itemize}
The rest of this paper presents the syntax in
Section~\ref{sec:syntax}, the (operational) semantics and the notion
of run-time errors in Section~\ref{sec:semantics}, the type system and
some examples in Section~\ref{sec:type-system}, and the main result in
Section~\ref{sec:results}.
We conclude, in Section~\ref{sec:further-work}, discussing an
alternative model for the semantics and pointing directions for future
work.


\begin{figure}
  \begin{equation*}
    \begin{aligned}
      e & \grmeq                   & & \text{\emph{Expressions}}\\
      & \grmin v                 & & \text{value}\\
      & \grmor \async {\vec v} e & & \text{fork activity}\\
      & \grmor \newClock         & & \text{new clock}\\
      & \grmor \drop v           & & \text{deregister from } v\\
      & \grmor \finish e         & & \text{wait to terminate}\\
      & \grmor \nextk            & & \text{advance phase}\\
      & \grmor \resume v         & & \text{ready to advance on } v\\
      & \grmor \lets x e e       & & \text{local declaration}\\
    \end{aligned}
    \qquad \qquad
    \begin{aligned}
      v & \grmeq                                  & & \text{\emph{Values}}\\
      & \grmin x                 & & \text{variable}\\
      & \grmor \unitV                           & & \text{unit}
     \\
      \\
      \\
      \\
      \\
     \\
     \\
   \end{aligned}
  \end{equation*}
  \caption{Top-level syntax of \xtenclocks}
  \label{fig:syntax}
\end{figure}


\section{Syntax}
\label{sec:syntax}

Object-oriented and type-safe, X10 provides for support for
concurrency, parallelism, and distribution.
Of particular interest to us is the finish synchronisation mechanism
that waits for the termination of parallel computations and the clock
primitive that allows forcing multiple concurrent activities to wait
for each other at certain points in time.

The top-level, or programmer's, language we address, \xtenclocks (X10
restricted to clocks and finish), is a subset of the X10 language,
generated by the grammar in Figure~\ref{fig:syntax}, and relies on a
base set of \emph{variables} ranged over by~$x$.
An \xtenclocks program is an expression~$e$ that can operate on
activities, clocks, or unit~$\unitV$ values.
To construct programs we compose expressions through the standard let
construct~$\lets x {e_1} {e_2}$, which binds variable~$x$ to the
result of expression~$e$ in the scope of expression~$e'$.

Below we present an example program with the purpose of illustrating
the syntax and informally presenting the semantics of the language.
The example is composed of three activities an outermost activity
$a_1$, defined from line~2 to line~16, an inner activity $a_2$,
spawned at line~4, and another inner activity $a_3$, spawned at line~5
and lasting until line~11.
Along the example we make use of the derived expression~$e_1;e_2$ that
abbreviates~$\lets x {e_1} {e_2}$, where $x$ not free in $e_2$.

\begin{lstlisting}[numbers=right]
  //activity a1
  finish
    let x = makeClock in (
      finish async e0; //activity a2
      async x ( //activity a3
        e1;
        resume x;
        e2;
        next;
        e3;
        drop x);
      e4;
      resume x;
      next;
      e5;
      drop x)
\end{lstlisting}

Activities can be \emph{registered} with zero or more clocks and may
share clocks with other activities.
A clock can thus count with zero or more different registered
activities, which are also called \emph{participants}.
When an activity~$a$ is registered with clock~$x$, we say that~$x$ is a
clock \emph{held} by~$a$.
Activities may only register themselves with clocks via two different
means: when they explicitly create a clock (line~3, \newClock{},
creates a clock and registers activity $a_1$ with the clock), and when
they inherit clocks from its parent activity in a spawning operation
(line~5 spawns activity~$a_3$ and registers it with the clock
associated with variable~$x$).
Expression~\lstinline|drop| deregisters an activity from a clock
(line~11, $\drop x$ deregisters activity~$a_3$ from clock~$x$,
whereas line~16 deregisters activity~$a_1$ from~$x$).
%
%
Activities are disallowed to manipulate clocks they are not
registered with.

In \xtenclocks, an activity synchronises via two different methods: by
waiting for every activity spawned by expression~$e$ to terminate
(line~4 only terminates when activity~$a_2$ and all its sub-activities
terminate), and by waiting for its held clocks to \emph{advance a
  phase} (activity~$a_3$ waits at line 9, whereas 
activity~$a_1$ waits at line~14). 
%
X10 distinguishes between \emph{local} and \emph{global} termination
of an expression.  Local termination of an expression corresponds to
concluding its evaluation (reducing to a value $v$).  An expression
terminates globally when it terminates locally and each activity
spawned by the expression has also terminated globally.
Expression~$\finish e$ converts the global termination of
expression~$e$ into a local termination (line~2 waits for the
expression in lines~3--16 to terminate, meaning that it waits as well
for the termination of activities $a_2$ and $a_3$; line~4 waits for
activity~$a_2$ (the result of the evaluation of~$\async {e_0} {} \!$)
to terminate before launching the sub-activity~$a_3$, in line~5).
The second method to synchronise activities is using clocks.
Groups of activities, defined by the participants of a clock, 
evaluate concurrently until they reach the end of a phase.
When every participant of the group reaches the end of the phase, then
all move to the next phase, while still executing concurrently.
Phases are delimited by expression~\nextk; activities evaluate this
expression to mark the end of a phase (activities $a_1$ and $a_3$ 
synchronise at lines 9 and 14).

An activity can inform all other participants of a
clock~$x$ that it has completed its phase by using an expression of
the form~$\resume x$, thus making clocks act as fuzzy
barriers~\cite{gupta:fuzzy-barrier} (line 7).
Expression~\lstinline|resume| can be viewed as an optimisation to diminish
contention upon advancing a phase: it allows activities blocked
on~\nextk{} to cease waiting for such activities (which can become
at most one phase behind the clock's phase). 
In the example, expression $e_2$ might execute at the same time as
expression $e_5$, since activity $a_3$ may trigger
(at line~5) activity $a_1$ to advance clock~$x$ (blocked at line~11),
thus evaluating expressions $e_2$ and $e_5$ in parallel.
If we omit expression~$\resume x$ from this example, then
expressions~$e_2$ and $e_5$ cannot evaluate in parallel.

Clocks can be implemented via a (sophisticated) $n$-ary synchronisation
mechanism that includes a natural number representing its \emph{global
  phase}, initially set to zero.
Advancing a clock's phase amounts to incrementing its global phase
when every registered activity has \emph{quiesced}; an activity is
quiescent on a clock~$x$ after performing a~$\resume x$.
An activity resumes all its held clocks together by
evaluating~\nextk{} and suspends itself until these clocks become
ready to advance to the next phase.



\section{Operational Semantics}
\label{sec:semantics}

This section describes the operation semantics of our language, a
possible execution of the example described in
Section~\ref{sec:syntax}, and the notion of run-time errors.

\paragraph{Operational Semantics}

\begin{figure}
  \begin{equation*}
    \begin{aligned}
      S & \grmeq \state H A         & & \text{\emph{States}}\\
      H & \grmeq \{c_1 \colon h_1, \dots, 
      c_n \colon h_n\}              & & \text{\emph{Heaps}}\\
      A & \grmeq \{l_1 \colon a_1,\dots,l_n \colon a_n\}
      & & \text{\emph{Sets of named activities}}\\
      h & \grmeq \heapClock p R Q   & & \text{\emph{Clock values}}\\
      R,Q & \grmeq \{l_1,\dots,l_n\}& & \text{\emph{Sets of activity names}}\\
      a & \grmeq \activity V e A    & & \text{\emph{Activities}} \\
      V & \grmeq \{c_1 \colon p_1,\dots, c_n \colon p_n\}
      & & \text{\emph{Clocks' local views}}
    \end{aligned}
    \qquad \qquad
    \begin{aligned}
      e & \grmeq \dots          & & \text{\emph{Expressions}}\\
      & \grmor \join l          & & \text{join activity}\\
      \\
      v & \grmeq \dots           & & \text{\emph{Values}}\\
      & \grmor c               & & \text{clock}\\
      \\
      p & \grmeq 0 \grmor 1 \grmor 2 \grmor \dots
      & & \text{\emph{Phases}}\\
    \end{aligned}
  \end{equation*}
  \caption{Run-time syntax of \xtenclocks}
  \label{fig:run-time-syntax}
\end{figure}


Figure~\ref{fig:run-time-syntax} depicts the run-time syntax of our
language.
The run-time system relies on one additional set, \emph{clock names}
(or heap addresses, since clocks are the only data structures we
allocate in the heap), ranged over by~$c$.
A state~$S$ of an \xtenclocks computation comprises a shared
heap~$H$ and a set of named activities~$A$ that run concurrently.
Activity names,~$l$, are taken from the set of variables introduced in
Section~\ref{sec:syntax}.
The heap stores clock values~$h$, triples comprising a natural
number~$p$ representing its global \textbf{p}hase, a set~$R$ with the \textbf
registered activities, and another set~$Q$ with the \textbf quiesced
activities.
These sets keep track of the activities that synchronise in the clock
($R$) and the activities that are ready to advance the clock to the
next phase ($Q$).
Set difference~$R\setminus Q$ identifies the activities yet to make
progress on a clock; the clock phase only advances when all activities
have quiesced (when~$R\setminus Q = \emptyset$, equivalently~$R = Q$).
The set of registered activities~$R$ also allows to enforce that an
activity only operates on registered clocks.

An activity~$a$ is composed of a set of clocks' local views~$V$, an
expression~$e$ under execution, and a set of sub-activities~$A$.
Each activity has its own perception of the global phase of a clock;
the clocks' local view~$V$ is a map from clock names to natural numbers
describing the local phase.
The global phase of a clock and that local to one of its activities
may diverge in case the activity issues a \resumek{} on the clock.
%
%
Only when the activity issues a \nextk, the local view of
the clock and the global phase become in sync.
At anytime, a clock's local view is at most one phase behind the
global phase.
Notice that an activity is itself a tree of activities, since each
activity holds a set of (named) sub-activities.
When evaluating an expression~$\finish e$, the 
activity starts sub-activities for evaluating expression~$e$ 
and all activities spawned by~$e$.
Otherwise, activities have no sub-activities.

We augment the syntax of expressions at run-time with~$\join l$.
%
%
Expression~$\join l$ results from evaluating~$\finish e$; label~$l$
identifies the activity that executes the body~$e$ of the~$\finishk$
expression and that produces the resulting value of the~$\finish e$
expression.


\begin{figure}[t]
  \begin{gather*}
    \tag \Rasync
    \frac{
      \{\vec c\} \subseteq \dom V 
      \qquad
      l' \text{ is fresh}
      \qquad
      Q' \eqdef \text{if } l \in Q \text{ then } Q \cup \{l'\} \text{ else } Q
    }{
      \state
      {H\update c {\heapClock p R Q}_{c \text{ in } \vec c}}
      {\activity V {\async {\vec c} e} {A}}
      \rightarrow_l
      H \update c {\heapClock {p} {R  \cup \{l'\}} {Q'}}_{c \text{ in } \vec c};
      \update {l'} {\activity {\update c {p}_{c \text{ in } \vec c}} {e} \emptyset};
      \activity V \unitV {A}
    }
    \\
    \tag \Rmake
    \frac{
      c \text{ is fresh}
    }{
      \state H {\activity V \newClock {A'}}
      \rightarrow_l
      H \update c {\heapClock 0 {\{l\}} \emptyset};
      \emptyset;
      \activity {V\update c 0} c {A'}
    }
    \\
    \tag \Rresume
    \frac{
      Q' \eqdef \text{if } p = V(c) \text{ then } Q\cup\{l\} \text{ else } Q
      \qquad
      l \notin Q 
    }{
      \state {H\update c {\heapClock p R Q}} {\activity V {\resume c} {A}}
      \rightarrow_l
      H\update c {\heapClock p R {Q'}};
      \emptyset;
      \activity V \unitV {A}
    }
    \\
    \tag \Rnext
    \frac{
      \begin{array}{c}
        C_1 \eqdef
        \{c \mid V(c) = p, H(c) = \heapClock {p} {R} {R}\}
        \quad
        C_2 \eqdef
        \{c \mid V(c) = p, H(c) = \heapClock {p + 1} {\_} {\_}\}
        \quad
        C_1 \cup C_2 = \dom V
      \end{array}
    }{
      \state H {\activity V \nextk {A}}
      \rightarrow_l
      H \update {c}{\heapClock{p + 1} {R} {\emptyset}}_{c \in C_1};
      \emptyset;
      \activity {\update{c} {V(c) + 1}_{c \in V}} \unitV {A}
    }
    \\
    \tag \Rdrop
    \frac{
      c \in\dom V
      \qquad
      H' \eqdef \text{if } R = \{l\} \text{ then } H \setminus \{c\} \text{
       else } H \update c {\heapClock p {R \setminus \{l\}} {Q \setminus
          \{l\}}}
    }{
      \state {H\update c {\heapClock p R Q}} {\activity V {\drop c} {A}}
      \rightarrow_l
      H';
     \emptyset;
      \activity {V \setminus \{c\}} \unitV {A}
    }
    \\
    \tag {\Rfinish}
    \frac{
      l_0 \text{ is fresh}
    }{
      \state H {\activity V {\finish e} {A}}
      \rightarrow_l
      H;
      \emptyset;
      \activity V {\join {l_0}} {A\update {l_0} {\activity V e \emptyset}}
    }
    \\
    \tag {\Rjoin}
    \state H
    {\activity V {\join l_0}
      {\{l_0\colon \activity \emptyset {v_0} \emptyset, \dots, 
        l_n\colon \activity \emptyset {v_n} \emptyset \}}}
    \rightarrow_l
    H;
    \emptyset;
    \activity \emptyset {v_0} \emptyset
  \end{gather*}
  \caption{Reduction rules for activities
    \hfill \fbox{$\state Ha \rightarrow_l H;A;a$}}
  \label{fig:reduction-rules-activity}
\end{figure}


\begin{figure}[t]
  \begin{gather*}
   \tag {\Rletval}
    \state H {A \update l {\activity V {\lets x v e} {A'}}}
    \rightarrow
    \state H {A \update l {\activity V {e \subst v x} {A'}}}
    \\
    \tag {\Rlet}
    \frac{
      l \in \dom H
      \qquad
      \state H {\activity V e {A}}
      \rightarrow_l
      {H'}; A'''; {\activity {V'} {e'} {A'}}
    }{
      \state H {A'' \update l {\activity V {\lets x e {e''}} {A}}}
      \rightarrow
      \state {H'} {A'' \update l {\activity {V'} {\lets x {e'} {e''}} {A'}}, A'''}
    }
    \\
    \tag {\Ractivity}
    \frac{
      \state H {A'}
      \rightarrow
      \state {H'} {A''}
    }{
      \state H {A \update l {\activity V e {A'}}}
      \rightarrow
      \state {H'} {A \update l {\activity V e {A''}}}
    }
  \end{gather*}
  \caption{Reduction rules for states
    \hfill \fbox{$\state HA \rightarrow \state HA$}}
  \label{fig:reduction-rules-states}
\end{figure}


We present the small step reduction rules for \xtenclocks{} in
Figures~\ref{fig:reduction-rules-activity}
and~\ref{fig:reduction-rules-states}.
Reduction for activities (Figure~\ref{fig:reduction-rules-activity}),
$H; a \rightarrow_l H'; A; a'$, operates on a heap~$H$ and an
activity~$a$, and produces a possible different heap~$H'$, a set~$A$
of activities spawned during the evaluation of the expression in~$a$,
and a new activity~$a'$. Label~$l$ is the name of activity~$a$.

Rule~$\Rasync$ is the only rule that affects the set of spawned
activities~$A$. The programmer specifies a list of clocks~$\vec c$ on
which the new activity is to be registered with.
The newly created activity (named~$l'$) is added to the set~$R$ of
activities registered with clocks~$\vec c$, and stored in the heap.
The result of spawning an activity is the unit value~$\unitV$.
The created activity is composed of a clock view holding, for
each clock~$c$ in~$\vec c$, a copy of the global phase~$p$, 
an expression~$e$ to be evaluated, and an empty set of sub-activities.
Syntax~$H\update{c}{h}$ describes a heap~$H'$ with a distinguished
entry~$c\colon h$; formally~$H'(c) = h$ if~$c = c'$ else~$H(c)$.
The new activity inherits each clock~$c$ quiescence property, 
\ie if~$l$ is quiescent on clock~$c$ so is~$l'$ ($Q' 
=\text{if }l \in Q \text{ then } Q\cup\{l'\} \text{ else } Q$).

The language specification reads ``Clocks are created using a factory
method on \texttt{x10.lang.Clock}. The current activity is
automatically registered with the newly created clock''~\cite[page
207]{saraswat:x10-report}.
Expression~$\newClock$ creates a new clock in the heap with
initial phase~$0$, with~$l$ as the only registered activity, and
with no resumed activities, $\heapClock 0 {\{l\}} \emptyset$.
The activity creating the clock maintains a local
clock view~$\{c\colon0\}$ stored in~$V$.
Regarding expression~$\resumek$ the language specification reads ``An
activity may wish to indicate that it has completed whatever work it
wishes to perform in the current phase of a clock \texttt{c} it is
registered with, without suspending altogether'' (page 208).
Rule~$\Rresume$ asserts that when the~$l$-labelled activity issues
a~$\resume c$, its label is recorded in the set of resumed
activities~$R$ if the clock local phase is in sync with the clock
global phase~($p = V(c)$);
otherwise, the effect of the expression is discarded~($p \neq V(c)$), 
since the clock has already advance to the next phase.
An activity may only perform a~$\resumek$ operation per clock per
phase~($l \notin Q$).

The language specification reads ``Execution of this statement
[\texttt{next}] blocks until all the clocks that the activity is
registered with (if any) have advanced. (The activity implicitly
issues a resume on all clocks it is registered with before
suspending.) [. . .]. An activity blocked on next resumes execution
once it is marked for progress by all the clocks it is registered
with'' (page 209).
In our model, for simplicity's sake, the programmer must issue
a~$\resumek$ on all held clocks before expression~$\nextk$.
As rule~$\Rnext$ states, expression~$\nextk$ blocks the activity until
all clocks have been resumed ($C_1$) or have already advance their
phases ($C_2$).
Notice that when activities are waiting on a clock~$c$, the clock can
be in one of three states: (a) there are non-quiescent activities on
the clock and~$c$ is neither a member of~$C_1$ nor of~$C_2$; (b) all
registered activities are quiescent on the clock, and so~$c$ is a
member of~$C_1$; (c) the clock has advanced to the next phase thus
becoming a member of~$C_2$.
When an activity advances a clock global phase, it stops being a member
of set~$C_1$ and becomes a member of set~$C_2$ for the remaining 
activities waiting on that clock.
Since rule~$\Rnext$ only updates the clock phase of those clocks
belonging to~$C_1$ ($H \update c {\heapClock {p + 1} R \emptyset}_{c
  \in C_1}$) it ensures that the global clock state is updated only
once.

The language specification reads ``An activity may drop a clock by
executing \texttt{c.drop()}. The activity is no longer considered
registered with this clock'' (page 209).
With expression~$\drop c$, the~$l$-labelled activity cedes its control
over clock~$c$: we remove~$c$ from clock view~$V$, and remove activity
identifier~$l$ from both sets~$R$ and~$Q$.
Two consequences of dropping a clock~$c$ are: a) activities waiting on
clock~$c$ are no longer blocked because of this activity; b) when
executing a~$\nextk$ expression, this activity no longer waits for
clock~$c$.
In case~$l$ is the only activity registered with clock~$c$, it is safe
to deallocate the clock, so that the clock's heap space can be
reclaimed without resorting to garbage collection.
The language specification reads ``An activity A executes finish S by
executing S and then waiting for all activities spawned by S (directly
or indirectly [. . .]) to terminate'' (page 196).
Expression~$\finish e$ creates a child activity and evaluates into
expression~$\join {l_0}$ (rule~$\Rfinish$), which in turn blocks while
there exist sub-activities running.
When all sub-activities have reduced to a value, activity~$l$ ($\join
{l_0}$) evaluates to the value in its sub-activity~$l_0$ and garbage
collects all other sub-activities (rule~$\Rjoin$).

The reduction for states (Figure~\ref{fig:reduction-rules-states}), $S
\rightarrow S'$, allows for non-deterministic choice of which activity
$l$ to evaluate (rule~$\Ractivity$), capturing the
concurrency present in X10 computations.
We evaluate the let binding from left-to-right (rule~$\Rlet$), when
the left-hand-side expression becomes a value, we substitute this
value for variable~$x$ in the continuation expression~$e$
(rule~$\Rletval$).

\paragraph{Example}

Recall the example from Section~\ref{sec:syntax}.
Consider a loading function that sets up the initial state from a
given expression, which in this case is~$S_0$, an empty heap and an
activity evaluating the code in the example under a dummy~$\letk$.
State~$S_0$ reduces in two steps, using rules~$\Rfinish$ and $\Rlet$.
The grey boxes highlight a redex and also the corresponding
contractum.
\begin{equation*}
S_0 = \emptyset; \{l_1 : \activity \emptyset {\lets z {\underbrace{\annotation{{\finish {\lets x \newClock {( \finish {(\asynck\ e_0)}; e_6)} }}}}_{\text{Example from Section~\ref{sec:syntax}}}} \unitV} \emptyset \}
\end{equation*}
where~$e_6$ is $(\async x {(e_1;\resume x;e_2;\nextk;e_3;\drop x)});
\resume x;e_4;\nextk;e_5;\drop x$.
We perform two reduction steps to illustrate the effect of
expression~$\finishk$ on the sub-activities of~$l_1$.
\begin{equation*}
 \emptyset;  \{l_1:
 (\emptyset,
 {\lets z {\annotation{\join {l_2}}} \unitV},
 {\{\annotation{
  l_2: \activity \emptyset {\lets x \newClock {( \finish {(\asynck\ e_0)}; e_6)}} \emptyset
 }\}})
\}
\end{equation*}

From this state on, while~$\joink$ remains blocked, we apply
rule~$\Ractivity$ to evaluate the child activities of~$l_1$.
We perform four further reduction steps~($\Ractivity$, $\Rlet$,
$\Rmake$, and $\Rletval$) and observe how expression~$\newClock$
updates the heap and the clock view of activity~$l_2$.
\begin{equation*}
{\{\annotation{c\colon\heapClock 0 {\{l_2\}} {\emptyset}}\}}
;
\{l_1\colon
 (
 \emptyset,
 {\lets z {\join {l_2}} \unitV},
%
 \{
  l_2\colon ( {\{\annotation{c\colon0}\}},
    {\lets x {\annotation c} {( \finish {(\asynck\ e_0)}; e_6)}}, \emptyset)
 \}
)
\}
\end{equation*}

The non-determinism of our semantics now allows for various different
reductions.
A possible outcome of a (multi-step) reduction is
\begin{align*}
\{c:\heapClock 0 {\{l_2,l_3\}} {\emptyset} \}
;
&
\{l_1: (\emptyset, {\lets z {\join {l_2}} \unitV},
\{
  l_2: \activity {\{c:0\}} {(\annotation{\resume c};\nextk;e_5;\drop c)}
  \emptyset
  ,\\ 
  & \phantom{\{l_1: (\emptyset, {\lets z {\join {l_2}} \unitV},\{} l_3: \activity {\{c:0\}} {(\resume c;e_2;\nextk;e_3;\drop c)}
  {\emptyset}
 \})\}
\end{align*}

We now illustrate the case when activities~$l_2$ and~$l_3$ are
evaluating expressions~$e_2$ and~$e_5$ concurrently.
\begin{align*}
&\Ractivity,\Rlet\rightarrow \Rresume \rightarrow_{l_2} \Ractivity, \Rletval \rightarrow
\\
&
\begin{aligned}
  \{c:\heapClock 0 {\{l_2,l_3\}} {\{\annotation{l_2}\}} \} ; \{l_1:
  (\emptyset, {\lets z {\join {l_2}} \unitV}, \{&l_2: \activity
  {\{c:0\}} {(\nextk;e_5;\drop c)} \emptyset ,
  \\
  & \hspace{-2em}
  l_3: \activity {\{c:0\}} {(\annotation{\resume
      c};e_2;\nextk;e_3;\drop c)} {\emptyset} \})\}
\end{aligned}
\\
&
\Ractivity,\Rlet\rightarrow \Rresume \rightarrow_{l_3}
\\
&
\begin{aligned}
  \{c:\heapClock 0 {\{l_2,l_3\}} {\{l_2,\annotation{l_3}\}} \} ;
  \{l_1: (\emptyset, {\join {l_2}}, \{& l_2: \activity
  {\{c: 0\}} {(\annotation{\nextk};e_5;\drop c)} \emptyset
  ,
\\
&
l_3: \activity {\{c:0\}} {(\annotation{\unitV};e_2;\nextk;e_3;\drop c)}
  {\emptyset} \})\}
\end{aligned}
\\
&\Ractivity,\Rlet\rightarrow \Rnext \rightarrow_{l_2}
\\
&
\begin{aligned}
  \{c:\heapClock 1 {\{l_2,l_3\}} {\annotation \emptyset} \} ; \{l_1:
  (\emptyset, {\lets z {\join {l_2}} \unitV}, \{& l_2: \activity
  {\{c:\annotation 1\}} {(\annotation{\unitV};e_5;\drop c)} \emptyset
  ,\\
  &l_3: \activity {\{c:0\}} {(\unitV;e_2;\nextk;e_3;\drop c)} {\emptyset}
  \})\}
\end{aligned}
\end{align*}

After activity~$l_3$ evaluates~$\resume c$, activity~$l_2$, which is
blocked evaluating~$\nextk$, progresses, thus allowing
expressions~$e_2$ and~$e_5$ to execute in parallel.
%
Notice that activity~$l_2$ remains in phase~0, while activity~$l_3$ is
in phase~1.

\newcommand{\rteSpc}{\ \ }
\begin{figure}
  \begin{align*}
    \tag{\Easync}
    \state H {A\update l {\activity V {\lets x {\async {\vec c} e} {e'}} {A'} }} \in\err
    &\rteSpc
    \text{ if } \vec c \not \subseteq \dom H \text{ or } c \not \in \dom V
    \\
    \tag{\Eresume}
    \state {H\update{c}{\heapClock {\_} {\_} Q}} {A\update l {\activity V {\lets x {\resume c} e} {A'} }}\in\err
    &\rteSpc
    \text{ if } l \in Q \text{ or } c \not \in \dom H \text{ or } c \not \in \dom V
    \\
    \tag{\Edrop}
    \state H {A\update l {\activity V {\lets x {\drop c} e} {A'} }}\in\err
    &\rteSpc
    \text{ if } c \not \in \dom H \text{ or } c \not \in \dom V
    \\
   \state H {A\update l {\activity V {\lets x \nextk e} {A'}}} \in\err &
    \rteSpc
    \text{ if } V(c) = p, H(c) = \activity p \_ Q, \text{
      and } 
    \\
    \tag{\EnextI}
    &\rteSpc \;\; l \not \in Q, \text{ for some } c
    \\
    \tag{\EnextII}
   \state H {A\update l {\activity V {\lets x \nextk e} {A'}}}  \in\err
   &\rteSpc
    \text{ if } c \in \dom V \text{ and } 
c \notin \dom H, \text{ for some } c
    \\
    \tag{\Eact}
    \state H {A\update l {\activity V {v} \_}}  \in \err &\quad\!\!
    \text{ if } V \neq \emptyset
    \\
    \tag{\EactSet}
    \frac{
      \state H {A'} \in \err
    }{
      \state H {A\update l {\activity \_ \_ {A'}}} \in \err
   }
\end{align*}
\caption{The set~$\err$ of run-time errors}
  \label{fig:run-time-errors}
\end{figure}


\paragraph{Run-time errors}

Run-time errors is the smallest set $\err$ of states generated by the
rules in Figure~\ref{fig:run-time-errors}. The notion is consistent
with all the conditions documented to raise exception
\texttt{ClockUseException}, as discussed in the X10 language
specification report~\cite{saraswat:x10-report}.
The type system we present in Section~\ref{sec:type-system} allow us
to reject, at compile time, programs that can potentially throw a
\texttt{ClockUseException}.
%

During an \asynck{} operation, an activity cannot transmit
unregistered clocks through its first argument
(rule~{\Easync}). Similarly, activities can only perform \resumek{} or
\dropk{} operations on clocks they are registered with
(rules~\Eresume{} and~\Edrop).
In particular, it constitutes an error for an activity to drop a clock
twice, or to resume a clock more than once (for the same phase) or
after dropping it.
We achieved a fine grained control over the clocks an activity is
registered with.
Specifically, it is possible to devise, at compile time, whether an
activity resumed or dropped all of its held clocks, as manifest from
our typing rules later.
Therefore, it constitutes an error when an activity evaluates a
$\nextk$ expression before resuming all its clocks (rule~$\EnextI$).
It is also an error when upon evaluating a~$\nextk$ there is a held
clock that is not in the heap (rule~$\EnextII$).
Furthermore, an activity cannot evaluate to a value without dropping
all of its clocks (rule~$\Eact$).
Rule~$\EactSet$ allows error propagation from sub-activities.

Earlier versions of the language specification report (until version
2.05) included two additional error conditions we quote:
\begin{itemize}
\item ``It is a static error if any activity has a potentially live
  execution path from a \texttt{resume} statement on a clock
  \texttt{c} to a async spawn statement (which registers the new
  activity on \texttt{c}) unless the path goes through a \texttt{next}
  statement'' (page 153, version 2.04). (See example 2,
  Section~\ref{sec:type-system}.);
\item ``While executing \texttt{S} [the body of a finish], an activity
  must not spawn any \texttt{clocked} asyncs. (Asyncs spawned during
  the execution of \texttt{S} may spawn \texttt{clocked} asyncs.)''
  (page 141, version 2.04).
\end{itemize}

Our type system guarantees soundness in presence of these two
conditions.


\begin{figure}[t]
   \begin{align*}
      \tau & \grmeq 
      \unitk \grmor \clockk(\alpha)
      & & \textit{Types}
    \end{align*}
 \caption{Syntax of types}
  \label{fig:syn-types}
\end{figure}


\begin{figure}[t]
 \begin{gather*}
    \tag{\Twfc, \Twfu}
    \mathcal R, \alpha \vdash \clockT \alpha
    \qquad
    \mathcal R \vdash \unitT
   \\
   \tag{\Tvar, \TclockR, \Tunit}
   \frac{
     \mathcal R \vdash \tau
   }{
     \Gamma,x\colon\tau; \mathcal R
     \vdash x \colon \tau 
   }
   \qquad
     \Gamma,c\colon\clockT \alpha; \mathcal R, \alpha
     \vdash c \colon \clockT \alpha 
   \qquad
   \Gamma; \mathcal R
   \vdash \unitV \colon \unitT
   \\
   \tag{\TclockSeq}
   \frac{
     \Gamma; \mathcal R \vdash v_1 \colon \clockT {\alpha_1} 
     \quad
     \cdots
     \quad
     \Gamma; \mathcal R \vdash v_n \colon \clockT {\alpha_n} 
     \qquad 
     \alpha_i \neq \alpha_j, \text{ if } i \neq j 
     \qquad
     \alpha_i \text{ not in } \Gamma
   }{
     \Gamma; \mathcal R \vdash v_1 \dots v_n \colon \{\alpha_1
     , \dots, \alpha_n\}
   }
 \end{gather*}
 \caption{Typing rules for values and for well-formed types}
 \label{fig:typing-values}
\end{figure}


\begin{figure}[t]
 \begin{gather*}
    \tag{\Tvalue, \Tmake}
    \frac{
      \Gamma; \mathcal R \vdash v \colon \tau
    }{
      \Gamma; \mathcal R; \mathcal Q
      \vdash v \colon (\tau, \mathcal R, \mathcal Q) 
    }
    \qquad
    \frac{
      \alpha \text{ is fresh}
    }{
      \Gamma; \mathcal R; \mathcal Q
      \vdash \newClock \colon (\clockT{\alpha}, \mathcal R \cup
      \{\alpha\}, \mathcal Q)
    }
    \\
   \tag{\Tresume, \Tdrop}
    \frac{
      \Gamma; \mathcal R
      \vdash v \colon \clockT \alpha
      \qquad
      \alpha \notin \mathcal Q
    }{
      \Gamma; \mathcal R; \mathcal Q
      \vdash \resume v \colon (\unitk, \mathcal R; \mathcal Q \cup \{\alpha\})
    }
    \qquad
    \frac{
      \Gamma; \mathcal R
      \vdash v \colon \clockT \alpha
    }{
      \Gamma; \mathcal R; \mathcal Q
      \vdash \drop v \colon (\unitk, \mathcal R \setminus \{\alpha\}, \mathcal Q \setminus \{\alpha\})
    }
    \\
   \tag{\Tasync,\Tnext}
   \frac{
     \Gamma; \mathcal R
     \vdash \vec v \colon \mathcal R'
     \qquad
     \Gamma; \mathcal R'; \mathcal {Q} \cap \mathcal R'
     \vdash e \colon (\_, \emptyset, \emptyset)
   }{
     \Gamma; \mathcal R; \mathcal Q
     \vdash \async {\vec v} e \colon (\unitk, \mathcal R; \mathcal Q)
   }
   \qquad
   \Gamma; \mathcal R; \mathcal R
    \vdash \nextk \colon (\unitk, \mathcal R; \emptyset)
    \\
   \tag{\Tfinish}
   \frac{
     \Gamma; \emptyset; \emptyset \vdash e \colon (\tau,
     \emptyset, \emptyset)
   }{
     \Gamma; \mathcal R; \mathcal Q
     \vdash \finish e \colon (\tau, \mathcal R; \mathcal Q)
   }
  \\
   \tag{\Tlet}
   \frac{
     \Gamma; \mathcal R; \mathcal Q \vdash e_1 \colon (\tau,
     \mathcal{R'}, \mathcal{Q'})
     \qquad
     \Gamma, x \colon \tau; \mathcal{R'}; \mathcal{Q'} \vdash
     e_2 \colon (\tau', \mathcal{R''}, \mathcal{Q''})
   }{
     \Gamma; \mathcal R; \mathcal Q
     \vdash \lets x {e_1} {e_2} \colon (\tau', \mathcal {R''},
     \mathcal {Q''})
   }
 \end{gather*}
 \caption{Typing rules for expressions}
 \label{fig:typing-expressions}
\end{figure}


\section{Type System}
\label{sec:type-system}

This section presents a type system that uses \emph{singleton types}
to track clock usage throughout a program.


For types we rely on an additional base set of singleton types ranged
over by $\alpha$.
The syntax of types, depicted in Figure~\ref{fig:syn-types},
introduces the type \unitT{} of unit values, and the type $\clockT
\alpha$ of a \emph{particular} clock. We assign a different type to
each clock in order to ensure the correct usage of the clock
constructs within a program.

The type system for \xtenclocks{} programs is defined in
Figures~\ref{fig:typing-values} and \ref{fig:typing-expressions}.
A typing~$\Gamma$ is a map from variables (or activity labels) and
clocks to types.
We write $\dom \Gamma$ for the domain of $\Gamma$. When $x \not \in
\dom \Gamma$ we write $\Gamma, x \colon \tau$ for the typing $\Gamma'$
such that $\dom \Gamma' = \dom \Gamma \cup \{x\}$, $\Gamma'(x) =
\tau$, and $\Gamma'(y) = \Gamma(y)$ for $y \neq x$.
%
%
The type system also uses sets of singleton types, ranged over by
$\mathcal R$, for \textbf registered clocks, and~$\mathcal Q$, for
\textbf quiescent clocks.


The typing rules for values and for well formed types
(Figure~\ref{fig:typing-values}) are simple to follow. Well-formed\-ness
for clock types (rule~\Twfc) ensures that activities only make use of
clocks they are registered with.
Rule~\TclockSeq{} ensures that different clocks (as those in the heap)
have distinct singleton clock types, a property that is crucial for
establishing type safety.
For typing expressions we use a type system
(Figure~\ref{fig:typing-expressions}) that records the changes made to
the set of registered clocks, either by creating or dropping clocks,
and to the set of quiescent clocks (using $\resumek$ and $\nextk$) of
an expression.
Typing judgements are of the form $\Gamma; \mathcal R; \mathcal Q
\vdash e \colon (\tau, \mathcal R', \mathcal Q')$ meaning that
expression $e$ is well typed assuming the types for the free
identifiers in $\Gamma$, the registered clocks in~$\mathcal R$, and
the quiescent clocks in~$\mathcal Q$.
The type of an expression is a triple recording the type $\tau$ of its
value, as well as the registered $\mathcal R'$ and the
quiescent~$\mathcal {Q'}$ sets after execution of the expression.

Most typing rules are straightforward.
When creating a clock (rule~\Tmake) we associate a new singleton type
$\alpha$ with the clock and include it in set of clocks registered by
the activity ($\mathcal R \cup \{\alpha\}$).
Rule~$\Tresume$, which asserts that ``an activity may invoke
\texttt{resume()} only on a clock it is registered with, and has not
yet dropped'' (page 208%
, \textit{vide} rule~\Tvar), marks clock~$\alpha$
as quiescent.
Notice that a clock cannot be resumed more than once for the same
phase~($\alpha \not \in \mathcal R$), contrary to the language
reference that reads, ``Nothing happens if the activity has already
invoked a resume on this clock in the current phase'' (page 208).
A~$\drop v$ expression removes clock~$v$ from both the sets~$\mathcal
R$ and~$\mathcal Q$, thus the clock cannot be passed to new
activities, be the target of a~$\resumek$ expression, or be dropped
again.

For expression~$\async {\vec c} e$, the language reference reads
``Starts a new activity, initially registered with clocks [$\vec v$],
and running [$e$]. The activity running this code must be registered
on those clocks'' (page 207, w.r.t.\@ rule~$\Tasync$).
Rule~$\Tasync$ asserts that when an activity spawns another
activity registered on a sequence of clocks, the quiescent property of
the clocks is preserved by propagating the information about the
quiescent clock $\alpha$ ($\mathcal {Q} \cap \mathcal R$).
Moreover, the new activity must have dropped all its clocks upon
termination, contrary to the language reference that reads, ``All
activities are automatically deregistered from all clocks they are
registered with on termination (normal or abrupt)'' (page 207).
An activity cannot share a clock it does not hold, as noted in the
language reference, ``lacking that registration, cannot register a
sub-activity on it [a clock] with async'' (page 208).
Expression $\nextk$ marks the end of a phase; it checks
that all clocks have been resumed and clears
the quiescent clocks for the new phase (rules~$\Tnext$).
%

The $\finishk$ construct may interfere with clocks and cause
programs to deadlock.
In order to avoid such situations we prevent the body of
a $\finish e$ expression ($e$) from accessing any clock already
defined, thus eliminating (nested) dependencies between clocks
and $\finishk$.
Rule~\Tfinish{} also forces $e$ to unregister from all clocks it has
created, and therefore $\finish e$ has no effect on registered and
quiescent clocks.
This follows the semantics of the current version of X10 that reads,
``Inside of \texttt{finish\{S\}}, all \texttt{clocked} asyncs must be
in the scope an unclocked async'' (page 210).
Refer to the examples below for further discussion on the deadlock 
problem.
When typing a $\letk$ expression (rule~\Tlet), its continuation~$e_2$
is typed taking into consideration the effects produced by
expression~$e_1$.
The type of the $\letk$ is that of~$e_2$, as usual.

We have deliberately deviated from the standard X10 semantics in three
cases: $\nextk$, $\dropk$, and $\resumek$. The reasons for such
deviation are: (a) to illustrate the power of singleton types in
keeping track of clocks, (b) to simplify the (operational and static)
semantics, (c) to enforce a programming discipline that may avoid
potential bugs, and (d) because the compiler has enough information to
suggest, or automatically insert, code fixes (\eg by enumerating the
clocks that need to be dropped before a \nextk; see examples below).

Below we discuss a few \xtenclocks programs and the semantic
guarantees our type system enforces.
We decorate the examples with the typing assumptions ($\Gamma,
\mathcal R, \mathcal Q$) holding for each expression.

\paragraph{Example 1: Aliasing}

Our first example concerns clock aliasing, only introduced in X10 in
version 2.01. The example may read a bit trivial but illustrates more
sophisticated aliasing situations, derived for example from procedure
calls. 
%
%
Clearly a type system with linear control like the one we are going to
present allows to relieve such a restriction.
%
%
\begin{lstlisting}[numbers=right]
  //activity a1
  let x = makeClock in (                  // {x:clock(alpha)},{alpha},emptyset
    async x ( //activity a2               // {x:lock(alpha)},{alpha},emptyset
       let y = x in (                     // {x:clock(alpha),y:clock(alpha)},{alpha},emptyset
         resume x;                        // {x:clock(alpha),y:clock(alpha)},{alpha},{alpha} 
         drop y)                          // {x:clock(alpha),y:clock(alpha)},emptyset,emptyset
    );                                    // {x:clock(alpha)},{alpha},emptyset 
    drop x)                               // emptyset,emptyset,emptyset
\end{lstlisting}
In our case the code is typable, assigning the same singleton type
$\clockT\alpha$ to both $x$ and $y$. Upon introducing variable
$x$~(line 2), it gets assigned type $\clockT \alpha$ (\textit{vide}
rules~\Tlet{} and~\Tmake). Variable~$y$~(line~4) gets assigned type
$\clockT \alpha$, the type of $x$~(again using rules~\Tlet{} and~\Tmake).

\paragraph{Example 2: Resume state inheritance}

Our second example deals with a restriction the language reference
disallowed up until version 2.04, ``A potentially live execution path
from a \texttt{resume} statement on a clock~\texttt{c} to an async
spawn statement'' (page 153, version 2.04).
%
%
Our operational semantics allows the forked activity to inherit the
\emph{resume state} (resumed/not resumed) of the parent activity
(\textit{vide} rule~$\Rasync$) and therefore preserves the quiescence
property of clocks and avoids a race condition on the clock. Notice
that the forked activity~$a_2$ inherits the quiescent state of clock
$\alpha$~(line~4), resumed at line~3.
%
%
\begin{lstlisting}[numbers=right]
  //activity a1
  let x = makeClock in (                  // {x:clock(alpha)},{alpha},emptyset
    resume x;                             // {x:clock(alpha)},{alpha},{alpha}
    async x ( //activity a2               // {x:clock(alpha)},{alpha},{alpha}
      next;                               // {x:clock(alpha)},{alpha},emptyset
      drop x                              // {x:clock(alpha)},emptyset,emptyset
     );                                   // {x:clock(alpha)},{alpha},{alpha}
    drop x)                               // {x:clock(alpha)},emptyset,emptyset
\end{lstlisting}

\paragraph{Example 3: Race condition generated by not inheriting the
  resume state}

Describes a race condition triggered by clock synchronisation.
Activity~$a_1$ creates a clock~$x$, starts a second activity $a_2$
registered with clock~$x$ that, in turn, resumes on~$x$ and starts a
third activity~$a_3$ also registered with~$x$.
\begin{lstlisting}[numbers=right]
  //activity a1
  let x = makeClock in (                  // {x:clock(alpha)},{alpha},emptyset 
    async x ( //activity a2               // {x:clock(alpha)},{alpha},emptyset
      resume x;                           // {x:clock(alpha)},{alpha},{alpha}
      async x ( //activity a3             // {x:clock(alpha)},{alpha},{alpha}
        next;                             // {x:clock(alpha)},{alpha},emptyset
        drop x                            // {x:clock(alpha)},emptyset,emptyset
       );                                 // {x:clock(alpha)},{alpha},{alpha}
       next;                              // {x:clock(alpha)},{alpha},emptyset
       drop x                             // {x:clock(alpha)},emptyset,emptyset
    );                                    // {x:clock(alpha)},{alpha},emptyset
    resume x;                             // {x:clock(alpha)},{alpha},{alpha}
    next;                                 // {x:clock(alpha)},{alpha},emptyset
    drop x)                               // {x:clock(alpha)},emptyset,emptyset
\end{lstlisting}
The race condition might occur because after $a_2$ resumes on $x$
(line~4), either activity~$a_1$ may advance clock~$x$ phase by executing
$\nextk$ (line~13) or $a_2$ may register a new activity~$a_3$
with~$x$ (line~5), blocking~$a_1$ until activity~$a_3$ executes
its $\nextk$ instruction (line~6).
By inheriting the resume status of clock~$x$, activity~$a_3$ does not
block activity~$a_1$ and the race condition disappears (\textit{vide}
rule~\Rasync{} in Figure~\ref{fig:reduction-rules-activity} and
rule~\Tasync{} in Figure~\ref{fig:typing-expressions}). 

\paragraph{Example 4: Resume after resume}

The next example deals with resuming after resuming, a pattern
accepted in X10. Rule~$\Tresume$ rejects the program below, since it
is able to determine that clock~$x$ is resumed twice.
%
\begin{lstlisting}[numbers=right]
  //activity a1
  let x = makeClock in (                  // {x:clock(alpha)},{alpha},emptyset
    resume x;                             // {x:clock(alpha)},{alpha},{alpha}
    resume x)                             // error:clock x already quiescent for activity a1
\end{lstlisting} 
%

\paragraph{Example 5: Explicit drops}

Unlike X10, we have decided to explicitly deregister activities from
clocks upon activity termination.
Our type system keeps track of the clocks an activity is registered
with, and rejects programs with activities that finish before
deregistering from all its clocks.
Clocks without registered activities can be safely garbage collected
(\textit{vide} rule~\Rdrop).
The following example fails to type check, since the launched
activity does not drop clock $x$.
The compiler may easily suggest an appropriate fix: adding a
\lstinline|drop x| after the \lstinline|next| instruction on line 5,
or even automatically introduce such an instruction.
\begin{lstlisting}[numbers=right]
  //activity a1
  let x = makeClock in (                  // {x:clock(alpha)},{alpha},emptyset
    async x (        //activity a2        // {x:clock(alpha)},{alpha},emptyset
      resume x;                           // {x:clock(alpha)},{alpha},{alpha}
      next                                // {x:clock(alpha)},{alpha},emptyset
    );                                    // error: activity a2 did not drop x 
    drop x)
\end{lstlisting}

\paragraph{Example 6: Finish/async deadlock}

Finally, we discuss the interplay among\finishk\!, $\asynck$, and
clocks, which may cause programs to deadlock. 
The following program deadlocks because activity $a_2$ is waiting on
$\nextk$ (line~6) for activity~$a_1$ to advance on $x$, which is
planned to occur at line~10, but $a_1$ is waiting on
$\finishk$ (line~3) for activity $a_2$ to terminate, so $a_1$ never
reaches line~10 and the program deadlocks.
\begin{lstlisting}[numbers=right]
  //activity a1
  let x = makeClock in (                  // {x:clock(alpha)},{alpha},emptyset
    finish                                // {x:clock(alpha)},emptyset,emptyset
      async x (        //activity a2      // error: clock x is not in scope of finish
        resume x; 
        next; 
        drop x
      ); 
    resume x;
    next;
    drop x)
\end{lstlisting}
The cause for deadlock is that activity $a_2$ is registered with a clock
that is defined outside the enclosing \finishk\!: clock $x$ is defined
in line~2, whereas the\finishk expression extends from line~3 to line~8.
Our type system rejects this program, because when typing a%
$\finish e$ expression we type check $e$ in an environment 
with no registered clocks (\textit{vide} rule~\Rfinish). 
%

\paragraph{Example 7: Clocked finish/clocked async}

Version 2.10 of the X10 language introduces keyword \lstinline{clocked} to
prefix async and finish.
\begin{quote}
  In the most common case of a single clock coordinating a few
  behaviors, X10 allows coding with an implicit clock.
  [...]
  A \lstinline{clocked finish} introduces a new clock.
  It executes its body in the usual way that a finish does---except
  that, when its body completes, the activity executing the clocked
  finish drops the clock, while it waits for asynchronous spawned
  asyncs to terminate.
  A \lstinline{clocked async} registers its async with the implicit clock
  of the surrounding clocked \lstinline{finish}.
  [...]
  Clocked finishes may be nested. The inner
  \lstinline{clocked finish} operates in a single phase of the outer one.
\end{quote}

This featured introduced in the language is purely ``syntactic
sugar,'' which makes the common practice of creating a single clock and
sharing it among activities simpler.
The following two code listings present a program with and without the
syntactic extension side-by-side.
On the left column we show an activity written on a language that
resembles X10~\cite{saraswat:x10-report}. Remember that, in X10,
expression~$\nextk$ implicitly issues a~$\resumek$ and also that
activities implicitly drop all clocks on exit.
On the right column we find an equivalent activity written in our language.
Activity $a_1$ spawns three activities $a_2$, $a_3$, and $a_4$.
Activities~$a_2$ and $a_3$ share the same clock, say~$c_1$,
whereas activity~$a_4$ holds a different clock, say~$c_2$, but
it does not hold~$c_1$.
In the first phase of clock~$c_1$ the system executes
concurrently~$e_1$, $e_3$, and activity~$a_4$.
In the second phase of clock~$c_1$, activity~$a_4$ has
terminated, and expressions~$e_2$ and $e_4$ execute concurrently.

\begin{tabular}{l|r}
\!\!\!\!\!
\begin{lstlisting}[boxpos=t]
//activity a1
clocked finish (

 clocked async ( //activity a2
  e1
  next;
  e2

 );
 clocked async ( //activity a3
  e3
  next;
  e4

 )
 clocked finish (


  clocked async ( //activity a4
   e5
   next;
   e6

  )

 )

)
\end{lstlisting}
\!\!
&
\begin{lstlisting}[boxpos=t]
//activity a1
finish                          //emptyset,emptyset,emptyset
 let c1 = makeClock in (        //{c1:clock(alpha1)},{alpha1},emptyset
  async c1 ( //activity a2      //{c1:clock(alpha1)},{alpha1},emptyset
   e1
   resume c1; next;             //{c1:clock(alpha1)},{alpha1},emptyset
   e2         
   drop c1;                     //{c1:clock(alpha1)},emptyset,emptyset
  );                            //{c1:clock(alpha1)},{alpha1},emptyset
  async c1 ( //activity a3      //{c1:clock(alpha1)},{alpha1},emptyset
   e3
   resume c1; next;             //{c1:clock(alpha1)},{alpha1},emptyset
   e4
   drop c1;                     //{c1:clock(alpha1)},emptyset,emptyset
  );                            //{c1:clock(alpha1)},{alpha1},emptyset
  finish                        //{c1:clock(alpha1)},emptyset,emptyset
   let c2 = makeClock in (      //{c1:clock(alpha1),
                                // c2:clock(alpha2)},{alpha2},emptyset
    async c2 (//activity a4     //{c1:clock(alpha1),
     e5                         // c2:clock(alpha2)},{alpha2},emptyset
     resume c2; next;           //{c1:clock(alpha1),
     e6                         // c2:clock(alpha2)},{alpha2},emptyset
     drop c2           //{c1:clock(alpha1),c2:clock(alpha2)},emptyset,emptyset
    );             //{c1:clock(alpha1),c2:clock(alpha2)},{alpha2}, emptyset
    drop c2            //{c1:clock(alpha1),c2:clock(alpha2)},emptyset,emptyset
   );                           //{c1:clock(alpha1)},{alpha1},emptyset
  drop c1                       //{c1:clock(alpha1)},emptyset,emptyset
 )                              //emptyset,emptyset,emptyset
\end{lstlisting}
\end{tabular}
~\\
The activity (on the right) obtained by expanding
\lstinline[morekeywords=clocked]{clocked finish} and
\lstinline[morekeywords=clocked]{clocked async} is still typable with
the type system we propose.


\section{Main results}
\label{sec:results}

This section is dedicated to the study of the main result
of our system, namely typing preservation and type safety for typable
programs.

We are only interested in well-formed states
(the rules in Figure~\ref{fig:well-formed-states} check whether a
state is well-formed). 
A state is well formed if for each clock, the set of registered
activities with the clock contains exactly those activities that can
manipulate it. 
State $$\{c: \heapClock \_ \emptyset \_\}; \{l \colon \activity
{\{c\colon \_\}} {\lets \_ \nextk \_} \_\}$$ is ill formed,
since activity $l$ uses clock $c$ and is not registered with $c$. The 
activity is able to advance $c$'s phase without becoming quiescent on $c$.
State $$\{c: \heapClock \_ {\{l, l', \dots\}} \_\}; \{l \colon
\activity \emptyset \_ \_, l' \colon \activity {\{c\colon \_\}} \_ \_,
\dots\}$$ is also ill formed, since activity $l$ is mentioned as
registered with clock $c$ and is not part of $l$'s local view (which
is $\emptyset$). Any other activity registered with $c$ ($l'$ in the
example) is bound to deadlock because $l$ will never quiesce on $c$.

\begin{figure}
  \begin{gather*}
    \tag{\WFactClock, \WFactClockE}
    \frac{
      c \vdash A \colon S_1
      \qquad
      c \vdash A' \colon S_2
    }{
      c \vdash A, \{ l \colon \activity V \_ {A'} \} \colon S_1 \cup
      S_2 \cup \dom V \cap \{c\}
    }
    \qquad
    \qquad
    c \vdash \emptyset \colon \emptyset
    \\
    \tag{\WFheap, \WFheapE}
    \frac{
      Q \subseteq S \subseteq \dom \Gamma 
      \qquad
      c \vdash A \colon S
      \qquad
      \Gamma; A \vdash H \colon \diamond
   }{
     \Gamma; A \vdash 
      H, \{ c \colon \heapClock \_ S Q \} \colon \diamond
    }
    \qquad
    \qquad
    \Gamma; A \vdash \emptyset \colon \diamond
    \\
    \tag{\WFactSet, \WFactSetE}
    \frac{
      H(c_i) = \heapClock \_ {\mathcal R_i} \_
      \qquad
      l \in \mathcal R_i
      \qquad
      H \vdash A \colon \diamond
      \qquad
      H \vdash A' \colon \diamond
    }{
      H \vdash 
      A, \{ l \colon {\activity {\{c_1\colon \_, \dots, c_n \colon \_\}} \_ {A'}}\} \colon \diamond
    }
    \qquad
    \qquad
    H \vdash \emptyset \colon \diamond
    \\
    \tag{\WFstate}
    \frac{
      H \vdash A \colon \diamond
      \qquad
      \Gamma; A \vdash H \colon \diamond
    }{
      \Gamma \vdash H; A \colon \diamond
    }
 \end{gather*}
  \caption{Well-formed states}
  \label{fig:well-formed-states}
\end{figure}


In order to state our results we must be able to type check run-time
expressions as well as machine states.
The typing rules for run-time expression $\join l$ and for machine
states and activities are depicted in
Figures~\ref{fig:typing-run-type-expressions} and
\ref{fig:typing-state-act}.
The type of a $\join l$ expression (rule~\Tjoin) is that of activity
$l$. Notice that $\join l$ is the result of evaluating a $\finish e$
expression~(rule~\Rfinish) that is, in fact, the type of
$e$~(rule~\Tfinish).
It is worth noticing that a heap $H$ is well typed if 
each clock is assigned to a different singleton type and if the clocks allocated
in the heap are exactly those of typing~$\Gamma$ (rule~\Theap{}
where~$\textbf C $ is the set of all clocks). 
Moreover, activities may only resume on registered clocks.
An activity $\activity V e A$ has the type of its expression $e$
(rule~\Tact), which must unregister from all its clocks before
terminating, since after evaluating $e$ it is expected that the set of
registered clocks should be empty.
Rule~\Tstate{} incorporates the definition of well-formed states into 
the type system.
The remaining typing rules should be easy to follow.

\begin{figure}[t]
 \begin{gather*}
  \tag{\Tjoin}
  \Gamma,l\colon\tau; \mathcal R; \mathcal Q
  \vdash \join l  \colon (\tau, \mathcal R, \mathcal Q)
\end{gather*}
 \caption{Typing rules for run-time expressions}
 \label{fig:typing-run-type-expressions}
\end{figure}


\begin{figure}[t]
  \begin{gather*}
    \tag{\TcheckView}
    \frac{
        \Gamma; \mathcal R  \vdash c_1 \dots c_n \colon \mathcal R
    }{
        \Gamma  \vdash \{c_1 \colon \_, \dots, c_n \colon \_\} \colon \mathcal R
    }
    \\
    \tag{\Tact}
    \frac{
      \Gamma \vdash V \colon \mathcal{R}
      \qquad
      \mathcal{Q} \subseteq \mathcal{R}
      \qquad
      \Gamma; \mathcal{R} ; \mathcal{Q} \vdash e \colon (\tau,
      \emptyset, \emptyset)
      \qquad
      \Gamma \vdash A
    }{
      \Gamma \vdash \activity V e A \colon \tau
    }
    \\
    \tag{\TactSet}
    \frac{
      \Gamma \vdash a_1 \colon \tau_1 \quad \cdots \quad 
      \Gamma \vdash a_n \colon \tau_n
    }{
      \Gamma, l_1 \colon \tau_1, \dots, l_n \colon \tau_n
      \vdash \{l_1\colon a_1, \dots, l_n \colon a_n \} 
    }
    \\
   \tag{\Theap}
    \frac{
      \Gamma; \mathcal R \vdash c_1 \dots c_n \colon \mathcal R
      \qquad
      \{c_1, \dots, c_n\} = \dom \Gamma \restrictedto{\textbf C}
    }{
      \Gamma \vdash \{c_1 \colon h_1, \dots, c_n \colon h_n\}
    }
    \\
    \tag{\Tstate}
    \frac{
      \Gamma \vdash H;A \colon \diamond
      \qquad
      \Gamma \vdash H
      \qquad
      \Gamma \vdash A
    }{
      \Gamma \vdash \state H A
    }
  \end{gather*}
  \caption{Typing rules for machine states}
  \label{fig:typing-state-act}
\end{figure}


\begin{lem}[Weakening]
  \label{lem:weakening}
  Let \emph{a} be a variable or a clock name.
  \begin{enumerate}
  \item If $H \vdash A \colon \diamond $ then $H,c\colon h \vdash A
    \colon \diamond$.
  \item If $\Gamma; \mathcal R \vdash v \colon \tau$ then $\Gamma, a
    \colon\tau' \vdash v \colon \tau$.
  \item If $\Gamma \vdash V\colon\mathcal R$ then $\Gamma,a\colon\tau
    \vdash V\colon\mathcal R$.
  \item If $\Gamma; \mathcal R; \mathcal Q \vdash e\colon T$ then
    $\Gamma,a\colon\tau; \mathcal R; \mathcal Q \vdash e\colon T$.
 \end{enumerate}
\end{lem}
\begin{proof}[Proof outline] 
  \begin{enumerate}
  \item By induction on the derivation of the typing
    rules. Case~$\WFactSetE$ is direct. For case~$\WFactSet$ we use
    the induction hypothesis to prove that $H,\update{c}{h} \vdash A
    \colon \diamond$ and $H,\update{c}{h} \vdash A' \colon \diamond$;
    the remaining conditions are given by the hypotheses.
  \item By inspecting the typing rules.
  \item We apply rule~$\TcheckView$ to typify the clocks of the view,
    then we prove $\TclockSeq$ with (2).
  \item By induction on the derivation of the typing relation.
    Cases~$\Tmake$, $\Tnext$, and $\Tjoin$ are direct.
    Case~$\Tresume$ and $\Tdrop$ are proved similarly, using (2) to typify
    clock~$v$.  The proof for cases~$\Tfinish$, $\Tasync$, and~$\Tlet$ 
    follow by induction hypothesis. For case~$\Tasync$ we also use rule~$\TclockSeq$
    and~(2) to typify the clocks of the arguments.
  \end{enumerate}

\end{proof}

Notice we do not allow heap weakening for it would introduce in the
type environment clocks not present in the state.


\begin{lem}[Substitution]
  \label{lem:substitution}
  If $\Gamma;\mathcal {R} \vdash v\colon \tau$ and $\Gamma,x\colon\tau;
  \mathcal R; \mathcal Q \vdash e \colon T$ 
  then $\Gamma; \mathcal R; \mathcal Q \vdash e\subst vx
  \colon T$.
\end{lem}
\begin{proof}[Proof outline] 
  For~$\Tvalue$ we analyse two cases: when the value is the variable
  being substituted, and when it is not replaced. For the
  former case, we apply Lemma~\ref{lem:weakening} on the first
  hypothesis. For the latter case we use rule~$\Tvar$.
  Rules~$\Tmake$, $\Tnext$, and $\Tjoin$ are direct.   
  Cases~$\Tdrop$ and~$\Tresume$ follow by induction hypothesis.
  Rule~$\Tasync$ is the most complex. For the clocks being shared we
  use Lemma~\ref{lem:weakening} and the second hypothesis. For the
  expression being spawned we apply the induction hypothesis.
  Rule~$\Tfinish$ is proved similarly to~$\Tasync$, but simpler, since 
  \Tfinish{} has no arguments and its set of clocks is empty.
\end{proof}

\begin{lem}[Preservation for activities]
  \label{lem:preservation-activities}
  If $\Gamma \vdash H$
  and $\Gamma \vdash V\colon \mathcal Q$
  and $\mathcal Q \subseteq \mathcal R$
  and $\Gamma; \mathcal R; \mathcal Q \vdash e\colon T$
  and $\Gamma \vdash A$
  and $l\in\dom H$
  and $\state H {\activity VeA} \rightarrow_l H'; A''; \activity {V'}{e'}{A'}$,
  then
 $\Gamma' \vdash H'$
  and $\Gamma' \vdash V'\colon \mathcal Q'$
  and $\mathcal Q' \subseteq \mathcal R'$
  and $\Gamma'; \mathcal R'; \mathcal Q' \vdash e'\colon T$
  and $\Gamma' \vdash A',A''$,
  for some $\Gamma' \supseteq \Gamma$.
\end{lem}
\begin{proof}[Proof outline] 
  Despite the scary look of the statement, its proof is a routine
  inspection of the rules in the relation $\state H {\activity VeA}
  \rightarrow_l H'; A''; \activity {V'}{e'}{A'}$.
\end{proof}

\begin{lem}[Preservation for \xtenclocks]
  \label{thm:preservation}
  If $\Gamma \vdash S$ and $S \rightarrow S'$ then $\Gamma' \vdash
  S'$ and $\Gamma \subseteq \Gamma'$.
\end{lem}
\begin{proof}[Proof outline] 
  By induction on the derivation of the relation $S\rightarrow S'$. In
  all cases we build the derivation tree for~$\Gamma \vdash S$ using
  rules \Tstate, \TactSet, and \Tact, collect the hypotheses, use the
  above Lemmas, and then build a tree for~$\Gamma' \vdash S'$ using
  the same typing rules.
  The base cases are when the derivation ends with rules \Rletval{}
  and \Rlet. For \Rletval{} we take $\Gamma'=\Gamma$ and use the
  substitution Lemma~\ref{lem:substitution}. For \Rlet{} we use the
  (specially crafted) preservation for activities
  (Lemma~\ref{lem:preservation-activities}), as well as the weakening
  (Lemma~\ref{lem:weakening}) for the extant activity set $A$.
  The induction step is when derivation ends with rule \Ractivity; in
  this case we use the Weakening Lemma.
\end{proof}



\begin{thm}[Type Safety]\label{thm:ts}
  If $\Gamma \vdash S$ and $S\rightarrow^* S'$, then $S' \nerr$.
\end{thm}
\begin{proof}[Proof outline]
  We first establish that $\Gamma' \vdash S'$ using preservation
  (Lemma~\ref{thm:preservation}). Then we proceed by contradiction.
  The contradiction is proved by induction on the definition of
  $\mathrm{Error}$ predicate. For the base cases of \resumek, \dropk,
  and \asynck{} we build the derivation trees for the errors in
  Figure~\ref{fig:run-time-errors}, to conclude that $\Gamma' \vdash V
  \colon \mathcal R$ and $\Gamma';\mathcal R \vdash c \colon
  \clockk(\alpha)$. Sequent $\Gamma' \vdash V \colon \mathcal R$ is
  derived from rule \TcheckView, which effectively establishes a
  \emph{one-to-one} correspondence between the clock names $c$ in $V$
  and the singleton types $\alpha$ in~$\mathcal R$. On the other hand,
  sequent $\Gamma';\mathcal R \vdash c \colon \clockk(\alpha)$ is
  derived via rule \Twfc, which says that $\alpha\in\mathcal R$. Since
  $\alpha\in\mathcal R$, the correspondence allows us to conclude that
  $c\in \dom V$.
  Establishing that $v\in\dom H$ is easier. Given that $\Gamma' \vdash
  H$, we conclude that $\dom H = \dom(\Gamma' \restrictedto{\textbf
    C})$, and from sequent $\Gamma';\mathcal R \vdash c \colon
  \clockk(\alpha)$ we know that $c\in\dom\Gamma'$, hence done.
%
\end{proof}



In Section~\ref{sec:semantics} we introduced a loading function that
builds the initial machine state corresponding to a given
expression. Such a state is typable if the expression is.
Let $\load(e)$ be defined as the state $\state \emptyset {\update l
  {\activity \emptyset {\lets x e \unitk} \emptyset}}$. Then we have:

\begin{lem}
  If $\Gamma; \emptyset; \emptyset \vdash e\colon T$ then $\Gamma
  \vdash \load (e)$.
\end{lem}

Our final result guarantees that a well-typed expression does
not reduce to an error.
\begin{cor}
  If $\Gamma; \emptyset; \emptyset \vdash e\colon T$ then $\load
  (e)$ does not reduce to an Error.
\end{cor}
\begin{proof}
  From the lemma above and Theorem~\ref{thm:ts}.
\end{proof}

We anticipate a \emph{progress property} for typable
processes. Typability ensures that processes do not get stuck when
dropping a clock that is not in its clock set anymore, or when
otherwise trying to access a clock that it not allocated in the
heap. The remaining case is \nextk{} where the activity waits for set
$C_1$ (the set of quiescent clocks the activity is registered with) to
grow until becoming (together with 
$C_2$---the set of clocks that have already advance their phase) the
clock set of the activity. And this is bound to happen for both
\nextk{} and \dropk{}, since in each activity both implicitly resume
all clocks.
%
We foresee as well that typability also rules out programs that
deadlock, since $\finishk$ expressions can only use 
clocks created in its body expressions.


\section{Discussion and future work}
\label{sec:further-work}


We study two synchronisation constructs of X10: a primitive finish
that waits for the termination of activities (lightweight threads),
and clocks (a generalisation of barriers).
To better understand the language we define an operational semantics
and a type system (alternative to the constraint-based
system~\cite{saraswat.jagadeesan:concurrent-clustered-programming})
for a subset of X10 called \xtenclocks.
Our main result is type safety for typable programs
(Theorem~\ref{thm:ts}).

Our semantics represents clocks in the heap as triples $\heapClock p R
Q$ relying on two sets for recording the registered activities~$R$
and the quiesced activities~$Q$ on a clock.
Implementing operations that work with sets is costly; for instance
rule~\Rnext{} needs to compute sets~$C_1$ and~$C_2$, by checking if
sets~$R$ and~$Q$ are equal, and then verify if $C_1 \cup C_2 = \dom
V$. Should we make a real life implementation of the proposed semantics,
set operations would have a significant impact on performance.
We sketch a much faster approach that chooses to represent clocks as
triples $\heapClock p r q$ describing the clock phase, as before, but
taking $r$ and $q$ as the cardinal numbers of sets~$R$ and~$Q$. With
this representation we lose information about the identity of the
activities registered with a clock and, in particular, we cannot
determine if an activity has already resumed in the current phase
(\textit{vide} rules~\Rasync{} and~\Rresume).
To overcome this problem we need to enrich the clock local view with
an indicator of whether an activity has performed a resume in the current
phase. Thus, a clock local view becomes a pair $\langle p, b\rangle$
containing the current clock phase~$p$ (as before) and the resume
boolean indicator $b$, describing when the activity has resumed.
With this information it is straightforward to adapt rules \Rasync,
\Rmake, \Rresume, \Rnext, and \Rdrop. For instance, rule
\Rresume{} only updates the clock global view ($q \leftarrow q + 1$)
whenever its local view indicator is false. Also, rule \Rnext{} needs
to set $r$ to zero when advancing the clock global phase, and to clear
the indicator $b$ upon advancing the clock local phase.
Checking that all activities registered with a clock have quiesced
amounts to compare two integer values ($r = s$), instead of two sets
$R$ and $Q$ as before.
The main reasons for not adopting the semantics just sketched are that
the chosen semantics needs fewer rules and is easier to read and
understand.



We intend to investigate imperative features of the language, 
specially those related with clocks, and also other language constructs.
The finish construct is not only used to wait for the termination of
sub-activities, but also, as the language reference reads, ``A
collection point for uncaught exceptions generated during the
execution of S [the body of a finish]''~\cite[page
196]{saraswat:x10-report}.
Enriching our model with exceptions seems like a natural, promising
follow-up of our work.
The language report also reads, ``X10 does not contain a register
statement that would allow an activity to discover a clock in a data
structure and register itself on it'' (page 208); we would like to
study type-safe extensions to the language that might alleviate this
restriction in controlled situations.
Furthermore, we expect to extend our results to \xtenclocks equipped
with recursion or some form of iteration.
Futures are a form of a function that evaluates asynchronously, like
an activity, but can be \emph{forced} to finish locally to return a
value.
The semantics of a future, in what regards termination, is like
the~$\finishk$ construct, but its use cases are different.
We would also like to allow futures to register themselves with
clocks, a feature missing in X10.

\emph{Phasers} are a coordination construct that unifies collective
and point-to-point synchronisations with performance results
competitive to existing barrier
implementations~\cite{shirako.peixotto.etal:phasers}.
Phasers can be seen as an extension over clocks that allow for more
fine-grained control over synchronisation modes.
\emph{Phaser accumulators} are reduction constructs for dynamic
parallelism that integrate with
phasers~\cite{shirako.peixotto.etal:phaser-accumulators}.
Although further investigation is needed, we believe our work can be
extended to accommodate phasers and phaser accumulators, specially
with regards to the operational similarities between clocks and
phasers.



\section*{Acknowledgements}

The authors would like thank the anonymous referees for constructive
criticisms and detailed comments.

\bibliography{x10}
\bibliographystyle{eptcs}

\end{document}